\documentclass[12pt]{article}
\usepackage[]{amsmath,amssymb}
\usepackage{amscd}
\usepackage{latexsym}
\usepackage{cite}
\usepackage{amsthm}

\newtheorem{definition}{Definition}[section]
\newtheorem{theorem}[definition]{Theorem}
\newtheorem{lemma}[definition]{Lemma}

\newtheorem{example}[definition]{Example}
\newtheorem{conjecture}[definition]{Conjecture}
\newtheorem{problem}[definition]{Problem}
\newtheorem{note}[definition]{Note}

\newtheorem{proposition}[definition]{Proposition}

\begin{document}
\title{\bf  
An action of the free product
$\mathbb Z_2 \star \mathbb Z_2 \star \mathbb Z_2$ on the
$q$-Onsager algebra and its current algebra
}
%Tatsuro Ito\footnote{Supported in part by JSPS grant
%18340022.} $\;$   and
\author{
Paul Terwilliger 
}
\date{}
%\footnote{This author gratefully acknowledges 
%support from the FY2007 JSPS Invitation Fellowship Program
%for Reseach in Japan (Long-Term), grant L-07512.}
%}
%\date{}
%to get date printout, comment out above line

\maketitle
\begin{abstract}
Recently Pascal Baseilhac and Stefan Kolb introduced some
automorphisms
$T_0$, $T_1$ of the $q$-Onsager algebra $\mathcal O_q$,
that are roughly analogous to the Lusztig automorphisms of
$U_q(\widehat{\mathfrak{sl}}_2)$.
We use $T_0, T_1$ and a certain antiautomorphism of $\mathcal O_q$
to obtain an action of the free product
$\mathbb Z_2 \star \mathbb Z_2 \star \mathbb Z_2$ on $\mathcal O_q$
as a group of (auto/antiauto)-morphisms.
The action forms a pattern much more
symmetric than expected.
We show that a similar
phenomenon occurs for the 
associated current algebra $\mathcal A_q$.
We give some conjectures and problems concerning $\mathcal O_q$ and
$\mathcal A_q$.


 \end{abstract}

\section{Introduction}

\noindent 
We will be discussing the $q$-Onsager algebra $\mathcal O_q$ 
\cite{bas1,qSerre}.
This algebra is infinite-dimensional and noncommutative,
with a presentation involving two generators and two
relations called the $q$-Dolan/Grady
relations.
The algebra appears in
a number of contexts which we now summarize.
The algebra $\mathcal O_q$ is a $q$-deformation of the Onsager
algebra from mathematical physics 
\cite{Onsager},
\cite[Remark~9.1]{madrid}
and is currently being
used to investigate statistical
mechanical models such as the XXZ open spin chain 
\cite{bas2,bas1,
basXXZ,
%%%BasBel,
BK05,
bas4,
basKoi}.
The algebra $\mathcal O_q$ appears in the theory of tridiagonal pairs;
this is a pair of diagonalizable 
linear transformations on a finite-dimensional
vector space, each acting on the eigenspaces of the other
in a block-tridiagonal fashion 
\cite{TD00,
LS99}.
A tridiagonal pair of $q$-Racah
type
\cite{ItoTer}
is essentially the same thing as a finite-dimensional irreducible
$\mathcal O_q$-module
 \cite[Theorem~3.10]{qSerre}.
See
\cite{INT,
TD00,
ItoTer,
IT:aug,
LS99,
madrid,
aw}
for work relating $\mathcal O_q$
and tridiagonal pairs. 
The algebra $\mathcal O_q$ comes up in algebraic combinatorics, in
connection with the subconstituent
algebra of a $Q$-polynomial distance-regular graph
\cite{tersub3,TD00}.
This topic is where $\mathcal O_q$ originated;
 to our knowledge the $q$-Dolan/Grady relations first appeared in
\cite[Lemma~5.4]{tersub3}.
The algebra $\mathcal O_q$ appears in the theory of quantum groups,
as a coideal subalgebra of
$U_q(\widehat{\mathfrak{sl}}_2)$ 
\cite{bc,bas8,kolb}.
There exists an injective algebra homomorphism from $\mathcal O_q$ into
the algebra $\square_q$
\cite[Proposition~5.6]{pospart},
and a noninjective algebra homomorphism from $\mathcal O_q$ into
the universal Askey-Wilson algebra
\cite[Sections~9,~10]{uaw}, \cite{pbw}.
\medskip

\noindent We will be discussing some automorphisms and antiautomorphisms
of $\mathcal O_q$.
In \cite{BK} Pascal Baseilhac and Stefan Kolb
introduced two automorphisms $T_0, T_1$ of $\mathcal O_q$
that are roughly analogous to the Lusztig automorphisms of
$U_q(\widehat{\mathfrak{sl}}_2)$.
More information about $T_0, T_1$ is given in
\cite{lusztigaut}.
Using $T_0, T_1$ and a certain antiautomorphism of $\mathcal O_q$,
we will obtain an action of the free product
$\mathbb Z_2 \star \mathbb Z_2 \star \mathbb Z_2$ on
$\mathcal O_q$ as a group of
(auto/antiauto)-morphisms. The action seems remarkable 
because it forms a pattern much more symmetric than expected.
We show that a similar phenomenon occurs for
the current algebra 
$\mathcal A_q$ of $\mathcal O_q$. Our main results are
Theorem
\ref{thm:z3}
and Theorem
\ref{thm:secondres}.
At the end of the paper we give some conjectures and problems
concerning $\mathcal O_q$ and $\mathcal A_q$.

\section{The $q$-Onsager algebra $\mathcal O_q$}

\noindent 
We will define the $q$-Onsager algebra after a few
comments.
Let $\mathbb F$ denote a field.
All vector spaces
%%and tensor products 
discussed in this paper are over $\mathbb F$.
All algebras discussed in this paper 
are
associative, over $\mathbb F$,  and have a multiplicative identity.
A subalgebra has the same multiplicative identity
as the parent algebra. 
For an algebra $\mathcal A$, 
an {\it automorphism} of  $\mathcal A$ is an 
algebra
isomorphism $\mathcal A \to \mathcal A$.
An {\it antiautomorphism}
of $\mathcal A$ is an $\mathbb F$-linear bijection
$\sigma : \mathcal A\to \mathcal A$ such that
$(XY)^\sigma= Y^\sigma X^\sigma$ for all $X,Y \in \mathcal A$.
If $\mathcal A$ is commutative, then there is no difference between
an automorphism and antiautomorphism of $\mathcal A$.
If $\mathcal A$ is noncommutative, then no map is both
an
automorphism and antiautomorphism of $\mathcal A$.
Recall the natural numbers
$\mathbb N = \lbrace 0,1,2,\ldots \rbrace$.
Fix $0 \not=q \in \mathbb F$ that is not a root of unity. We will
use the notation
\begin{eqnarray*}
\lbrack n\rbrack_q = \frac{q^n-q^{-n}}{q-q^{-1}}
\qquad \qquad n \in \mathbb N.
\end{eqnarray*}
%\noindent All unadorned tensor products are over $\mathcal F$.

\begin{definition}
\label{def:qons}
{\rm (See \cite[Section~2]{bas1},
\cite[Definition~3.9]{qSerre}.)}
\rm Define the algebra ${\mathcal O}_q$ 
by generators $A,B$ and relations
\begin{eqnarray}
&&A^3B-\lbrack 3\rbrack_q A^2BA+ 
\lbrack 3\rbrack_q ABA^2 -BA^3 = (q^2-q^{-2})^2 (BA-AB),
\label{eq:dg1}
\\
&&B^3A-\lbrack 3\rbrack_q B^2AB + 
\lbrack 3\rbrack_q BAB^2 -AB^3 = (q^2-q^{-2})^2 (AB-BA).
\label{eq:dg2}
\end{eqnarray}
\noindent We call ${\mathcal O}_q$ the {\it  $q$-Onsager algebra}. 
The relations (\ref{eq:dg1}),  (\ref{eq:dg2}) 
are called the {\it $q$-Dolan/Grady relations}.
\end{definition}
\noindent 
%The algebra $\mathcal O_q$ is infinite-dimensional 
%and noncommutative \cite{BK}.
%\medskip

\noindent We now consider some automorphisms of $\mathcal O_q$.
By the form of the relations
(\ref{eq:dg1}),  (\ref{eq:dg2}) there exists an automorphism of
$\mathcal O_q$ that swaps $A, B$. The following
automorphisms of $\mathcal O_q$ are less obvious.
In \cite{BK} Pascal Baseilhac and Stefan Kolb
introduced some automorphisms $T_0$, $T_1$ of ${\mathcal O}_q$ that
satisfy
\begin{align}
&T_0(A)=A, \qquad \qquad
T_0(B) = B + \frac{qA^2B-(q+q^{-1})ABA+q^{-1}BA^2}{(q-q^{-1})(q^2-q^{-2})},
\label{eq:bk1}
\\
&T_1(B)=B, \qquad \qquad
T_1(A) = A + \frac{qB^2A-(q+q^{-1})BAB+q^{-1}AB^2}{(q-q^{-1})(q^2-q^{-2})}.
\label{eq:bk2}
\end{align}
The inverse automorphisms satisfy
\begin{align}
&T_0^{-1}(A)=A, \qquad \qquad
T_0^{-1}(B) = B + \frac{q^{-1}A^2B-(q+q^{-1})ABA+qBA^2}{(q-q^{-1})(q^2-q^{-2})},
\label{eq:bk3}
\\
&T_1^{-1}(B)=B, \qquad \qquad
T_1^{-1}(A) = A + \frac{q^{-1}B^2A-(q+q^{-1})BAB+qAB^2}{(q-q^{-1})(q^2-q^{-2})}.
\label{eq:bk4}
\end{align}

\noindent In \cite{BK} the automorphisms $T_0, T_1$ 
are used to construct a
Poincar\'e-Birkhoff-Witt (or PBW)
basis for $\mathcal O_q$. In that
construction the following result is used.
\begin{lemma}
\label{lem:qcom}
{\rm (See \cite[Lemma~3.1]{BK}.)}
For the algebra $\mathcal O_q$,
the map $T_0$ sends
\begin{align}
\label{eq:qcom1}
qBA-q^{-1}AB \mapsto qAB-q^{-1}BA,
\end{align}
and the map
$T_1$ sends
\begin{align}
\label{eq:qcom2}
qAB-q^{-1}BA \mapsto qBA-q^{-1}AB.
\end{align}
\end{lemma}
\begin{proof} 
The map $T_0$ is an automorphism of $\mathcal O_q$ that fixes $A$.
Therefore, $T_0$ sends 
$qBA-q^{-1}AB \mapsto 
qT_0(B)A-q^{-1}AT_0(B)$.
To check that 
$qT_0(B)A-q^{-1}AT_0(B) = 
qAB-q^{-1}BA$, 
eliminate $T_0(B)$ using
(\ref{eq:bk1})  and evaluate the result using
(\ref{eq:dg1}). We have verified the assertion about $T_0$. The
assertion about $T_1$ is similarly verified.
\end{proof}

\noindent
The automorphism group
${\rm Aut}(\mathcal O_q)$ consists of the
automorphisms of 
the algebra $\mathcal O_q$; the group operation is composition.
\begin{definition}
\label{def:N} \rm Let $N$ denote the subgroup of
${\rm Aut}(\mathcal O_q)$ generated by
 $T^{\pm 1}_0, T^{\pm 1}_1$.
 \end{definition}

\begin{lemma} 
\label{lem:free}
{\rm (See 
\cite[Section~1]{pbw}.)}
The group $N$ is freely generated by $T^{\pm 1}_0, T^{\pm 1}_1$.
\end{lemma}
%%%%%%%%%%%%%%%%%%%%
%%freely generate 
%%a subgroup of 
%%${\rm Aut}(\mathcal O_q)$.

\noindent We have been discussing automorphisms of $\mathcal O_q$.
We now bring in antiautomorphisms of $\mathcal O_q$.
\begin{lemma} 
\label{lem:gam1}
There exists an antiautomorphism $S$ of $\mathcal O_q$
that fixes $A$ and $B$. Moreover $S^2=1$.
\end{lemma}
\begin{proof} By the form of the $q$-Dolan/Grady relations.
\end{proof}
\noindent
The antiautomorphism $S$ is related to the automorphisms
$T_0$, $T_1$ in the following way.
\begin{lemma} 
\label{lem:alphabeta1}
For the algebra $\mathcal O_q$,
\begin{align}
S T_0 S = T^{-1}_0, \qquad \qquad
S T_1 S = T^{-1}_1.
\label{eq:gam}
\end{align}
\end{lemma}
\begin{proof} We verify
the equation on the left in
(\ref{eq:gam}). In that equation,
each side is
an
automorphism of $\mathcal O_q$.
These automorphisms agree at $A$ and $B$;
this is checked using
(\ref{eq:bk1}) and
(\ref{eq:bk3}). These automorphisms are equal
since $A$, $B$ generate $\mathcal O_q$.
We have verified the equation on the left in
(\ref{eq:gam}). The 
 equation on the right in
(\ref{eq:gam}) is similarly verified.
\end{proof}

\noindent Let ${\rm AAut}(\mathcal O_q)$ denote the group consisting
of the automorphisms and antiautomorphisms of $\mathcal O_q$;
the group operation is composition.
The group ${\rm Aut}(\mathcal O_q)$ is a normal subgroup of
 ${\rm AAut}(\mathcal O_q)$ with index 2. 
An element of ${\rm AAut}(\mathcal O_q)$ will be called 
an {\it (auto/antiauto)-morphism of $\mathcal O_q$}.

\begin{definition}
\label{def:GH}
\rm 
Let $H$ denote the subgroup of 
${\rm AAut}(\mathcal O_q)$ generated by $S$.
Let $G$ denote the subgroup of 
${\rm AAut}(\mathcal O_q)$ generated by $H$ and $N$.
\end{definition}

\begin{lemma} The following {\rm (i)--(iv)} hold.
\begin{enumerate}
\item[\rm (i)] 
The group $H$ has order 2 and is not contained in $N$.
\item[\rm (ii)] 
The group $N$ is a normal subgroup of $G$
with index 2.
\item[\rm (iii)] 
$G=
N \rtimes H$ (semidirect product).
\item[\rm (iv)] 
$N = {\rm Aut}(\mathcal O_q) \cap G$.
\end{enumerate}
\end{lemma}
\begin{proof} (i) The group $H$ has order 2 by the last assertion of
Lemma
\ref{lem:gam1}.
The group $H$ 
is not contained in $N$, since 
${\rm Aut}(\mathcal O_q)$ contains $N$ but not $S$.
\\
\noindent (ii) By Lemma
\ref{lem:alphabeta1} and part (i) above.
\\
\noindent (iii) By parts (i), (ii) above.
\\
\noindent (iv) The group $G$ is the union of cosets
$N$ and $N S$.
The elements of $N$ are in 
${\rm Aut}(\mathcal O_q)$, and the elements of $NS$ are not in
${\rm Aut}(\mathcal O_q)$.
\end{proof}

\noindent We now consider $G$ from another point of view.
Let $\mathbb Z_2$
denote the group with two elements. The free
product 
$\mathbb Z_2 \star \mathbb Z_2 \star \mathbb Z_2$
has a presentation by generators $a,b,c$ and relations
$a^2=b^2=c^2=1$.
Shortly we will display a group isomorphism
$\mathbb Z_2 \star \mathbb Z_2 \star \mathbb Z_2 \to 
G$. To motivate this isomorphism we give a second presentation
of 
$\mathbb Z_2 \star \mathbb Z_2 \star \mathbb Z_2$ by
generators and relations.

\begin{lemma}
\label{lem:isoiso}
The group
$\mathbb Z_2 \star \mathbb Z_2 \star \mathbb Z_2$
 is isomorphic to the group defined by generators
$s$, $t^{\pm 1}_0$, $t^{\pm 1}_1$ and relations
\begin{align}
&t_0 t^{-1}_0 = t^{-1}_0 t_0 = 1, \qquad \qquad
t_1 t^{-1}_1 = t^{-1}_1 t_1 = 1,
\\
&s^2=1, \qquad \qquad s t_0 s = t^{-1}_0, \qquad \qquad s t_1 s = t^{-1}_1.
\label{eq:rels}
\end{align}
An isomorphism sends
\begin{align}
\label{eq:send1}
a \mapsto s t_1,
\qquad \qquad
b \mapsto t_0 s,
\qquad \qquad
c \mapsto s.
\end{align}
The inverse isomorphism sends
\begin{align}
\label{eq:send2}
t_0 \mapsto bc, \qquad 
t^{-1}_0 \mapsto cb, \qquad 
t_1 \mapsto ca, \qquad 
t^{-1}_1 \mapsto ac, \qquad 
s \mapsto  c.
\end{align}
\end{lemma}
\begin{proof} One checks that each map is a group homomorphism and
the maps are inverses. Consequently each map is a group isomorphism.
\end{proof}

\begin{proposition} 
\label{prop:ZZZisoG}
There exists a group isomorphism 
$\mathbb Z_2 \star \mathbb Z_2 \star \mathbb Z_2 \to G$ that sends
\begin{align} 
a \mapsto ST_1, \qquad \qquad b\mapsto T_0 S, \qquad \qquad c\mapsto S. 
\label{eq:target1}
\end{align}
The inverse isomorphism sends
\begin{align}
T_0 \mapsto bc, \qquad 
T^{-1}_0 \mapsto cb, \qquad 
T_1 \mapsto ca, \qquad 
T^{-1}_1 \mapsto ac, \qquad 
S \mapsto  c.
\label{eq:target2}
\end{align}
\end{proposition}
\begin{proof} For notational convenience we identify the group
$\mathbb Z_2 \star \mathbb Z_2 \star \mathbb Z_2$ with the group
defined in Lemma
\ref{lem:isoiso}, via the isomorphism in
Lemma
\ref{lem:isoiso}.
Comparing 
(\ref{eq:rels}) 
with the relations in
Lemmas
\ref{lem:gam1},
\ref{lem:alphabeta1} we obtain
a surjective group homomorphism 
$\gamma : 
\mathbb Z_2 \star \mathbb Z_2 \star \mathbb Z_2 
\to G$ that sends $s\mapsto S$ and
$t^{\pm 1}_0 \mapsto T^{\pm 1}_0 $ and $t^{\pm 1}_1 \mapsto T^{\pm 1}_1$.
Using the identification
(\ref{eq:send1})
we find that $\gamma$ acts as in
(\ref{eq:target1}).
We show that $\gamma$ is an isomorphism.
Let $\mathcal N$ denote the subgroup of
$\mathbb Z_2 \star \mathbb Z_2 \star \mathbb Z_2 $
 generated by
$t^{\pm 1}_0, t^{\pm 1}_1$.  
From the relations (\ref{eq:rels}) we see that 
$\mathbb Z_2 \star \mathbb Z_2 \star \mathbb Z_2 $
is the union of $\mathcal N$ and $\mathcal N s$.
We have $\gamma (\mathcal N) = N$
and 
$\gamma (\mathcal N s) = N S$. The cosets $N$, $NS$ are disjoint
and $N$ contains the identity, so 
the kernel of $\gamma$ is contained in $\mathcal N$.
This kernel is the identity  
 by Lemma
\ref{lem:free}.
Therefore $\gamma$ is injective and hence an isomorphism.
Line (\ref{eq:target2}) follows from
(\ref{eq:send2}).
\end{proof}

\noindent  We now give our first main result.
For notational convenience define
\begin{align}
C = \frac{q^{-1}BA-qAB}{q^2-q^{-2}}.
\label{eq:C}
\end{align}

\begin{theorem}
\label{thm:z3}
The free product $\mathbb Z_2 \star \mathbb Z_2 \star \mathbb Z_2$
acts on the algebra 
$\mathcal O_q$ as a group of (auto/antiauto)-morphisms
in the following way.
\begin{enumerate}
\item[\rm (i)] The generator $a$ acts as an antiautomorphism 
 that sends
\begin{align}
A \mapsto A + \frac{BC-CB}{q-q^{-1}}, \qquad \qquad 
B \mapsto B, \qquad \qquad C \mapsto C.
\label{eq:actA}
\end{align}
\item[\rm (ii)] The generator $b$ acts as an antiautomorphism that sends
\begin{align}
B \mapsto B + \frac{CA-AC}{q-q^{-1}},\qquad \qquad 
C \mapsto C, \qquad \qquad A \mapsto A.
\label{eq:actB}
\end{align}
\item[\rm (iii)] The generator $c$ acts as an antiautomorphism
that sends
\begin{align}
C \mapsto C + \frac{AB-BA}{q-q^{-1}},\qquad \qquad 
A \mapsto A, \qquad \qquad B \mapsto B.
\label{eq:actC}
\end{align}
\item[\rm (iv)] On $\mathcal O_q$,
\begin{align}
&a= ST_1=T^{-1}_1 S, \qquad \qquad b= T_0 S=ST^{-1}_0, \qquad \qquad c=S,
\label{eq:part1}
\\
&T_0 = bc, \;\quad \qquad  
T^{-1}_0 = cb, \;\quad \qquad 
T_1 = ca, \;\quad \qquad 
T^{-1}_1 = ac.
\label{eq:part2}
\end{align}
\end{enumerate}
\noindent The above
 $\mathbb Z_2 \star \mathbb Z_2 \star \mathbb Z_2$
action is faithful.
\end{theorem}
\begin{proof} By 
Lemma
\ref{lem:alphabeta1}
and Proposition
\ref{prop:ZZZisoG}, the group
 $\mathbb Z_2 \star \mathbb Z_2 \star \mathbb Z_2$ acts
 faithfully on
 $\mathcal O_q$ as a group of
 (auto/antiauto)-morphisms in a way that satisfies
(\ref{eq:part1}),
(\ref{eq:part2}).
By 
(\ref{eq:part1}),
each of $a,b,c$
acts on $\mathcal O_q$ as an antiautomorphism.
Their actions
(\ref{eq:actA})--(\ref{eq:actC})
are routinely obtained using 
(\ref{eq:bk1})--(\ref{eq:bk4})  and
(\ref{eq:C}),
along with
Lemmas
\ref{lem:qcom},
\ref{lem:gam1}.
%%%%%%%%%%%%%\ref{lem:alphabeta1}.
\end{proof}

\begin{note} \rm Motivated by Theorem
\ref{thm:z3}, one might conjecture that
there exists an automorphism of $\mathcal O_q$ that
sends $A\mapsto B\mapsto C \mapsto A$. This conjecture is false,
since $A,B$ satisfy the 
$q$-Dolan/Grady relations and
$B,C$ do not. This last assertion can be checked by considering
the actions of $B,C$ on the 4-dimensional $\mathcal O_q$-module given in
the proof of \cite[Lemma~10.8]{uaw}.
\end{note}

\section{The current algebra $\mathcal A_q$}

\noindent 
In the previous section we obtained an action of
$\mathbb Z_2 \star \mathbb Z_2 \star \mathbb Z_2$ on the
$q$-Onsager algebra $\mathcal O_q$. In this section we 
do something similar for the corresponding
current algebra $\mathcal A_q$.
In \cite{BK05} Baseilhac and Koizumi 
introduce $\mathcal A_q$ in order to solve boundary integrable systems
with hidden symmetries related to a coideal subalgebra of
$U_q(\widehat{\mathfrak{sl}}_2)$.  
In
\cite[Definition~3.1]{basnc} Baseilhac and K. Shigechi
give a presentation of  $\mathcal A_q$
by generators and relations.
The generators are 
denoted 
$\mathcal W_{-k}$, $\mathcal W_{k+1}$,  $\mathcal G_{k+1}$, 
 $\mathcal{\tilde G}_{k+1}$,
where $k \in \mathbb N$.
In \cite[Lemma~2.1]{basBel}, Baseilhac and S. Belliard display
some central elements 
$\lbrace \Delta_{k+1}\rbrace_{k \in \mathbb N}$ for 
$\mathcal A_q$.
In \cite[Corollary~3.1]{basBel}, it is shown that $\mathcal A_q$ is generated by
these central elements together with
 $\mathcal W_0, 
\mathcal W_1$.
The elements $\mathcal W_0,\mathcal W_1$
 satisfy the 
$q$-Dolan/Grady relations
%%%%(\ref{eq:dg1}), (\ref{eq:dg2}) 
\cite[eqn.~(3.7)]{basBel}.
In \cite[Conjecture~2]{basBel} it is conjectured 
that $\mathcal O_q$ is a homomorphic image of $\mathcal A_q$.
We now recall the definition of $\mathcal A_q$.
\begin{definition}\rm
\label{def:Aq}
(See 
\cite{BK05}, \cite[Definition~3.1]{basnc}.)
Define the algebra $\mathcal A_q$
by generators
\begin{align}
\label{eq:4gens}
 \mathcal W_{-k}, \qquad  \mathcal W_{k+1},\qquad  
 \mathcal G_{k+1},
\qquad
\mathcal {\tilde G}_{k+1}, \qquad \qquad  k \in \mathbb N
\end{align}
 and relations
\begin{eqnarray}
&&
 \lbrack \mathcal W_0, \mathcal W_{k+1}\rbrack= 
\lbrack \mathcal W_{-k}, \mathcal W_{1}\rbrack=
({\mathcal{\tilde G}}_{k+1} - \mathcal G_{k+1})/(q+q^{-1}),
\label{eq:3p1}
\\
&&
\lbrack \mathcal W_0, \mathcal G_{k+1}\rbrack_q= 
\lbrack {\mathcal{\tilde G}}_{k+1}, \mathcal W_{0}\rbrack_q= 
\rho  \mathcal W_{-k-1}-\rho 
 \mathcal W_{k+1},
\label{eq:3p2}
\\
&&
\lbrack \mathcal G_{k+1}, \mathcal W_{1}\rbrack_q= 
\lbrack \mathcal W_{1}, {\mathcal {\tilde G}}_{k+1}\rbrack_q= 
\rho  \mathcal W_{k+2}-\rho 
 \mathcal W_{-k},
\label{eq:3p3}
\\
&&
\lbrack \mathcal W_{-k}, \mathcal W_{-\ell}\rbrack=0,  \qquad 
\lbrack \mathcal W_{k+1}, \mathcal W_{\ell+1}\rbrack= 0,
\label{eq:3p4}
\\
&&
\lbrack \mathcal W_{-k}, \mathcal W_{\ell+1}\rbrack+
\lbrack \mathcal W_{k+1}, \mathcal W_{-\ell}\rbrack= 0,
\label{eq:3p5}
\\
&&
\lbrack \mathcal W_{-k}, \mathcal G_{\ell+1}\rbrack+
\lbrack \mathcal G_{k+1}, \mathcal W_{-\ell}\rbrack= 0,
\label{eq:3p6}
\\
&&
\lbrack \mathcal W_{-k}, {\mathcal {\tilde G}}_{\ell+1}\rbrack+
\lbrack {\mathcal {\tilde G}}_{k+1}, \mathcal W_{-\ell}\rbrack= 0,
\label{eq:3p7}
\\
&&
\lbrack \mathcal W_{k+1}, \mathcal G_{\ell+1}\rbrack+
\lbrack \mathcal  G_{k+1}, \mathcal W_{\ell+1}\rbrack= 0,
\label{eq:3p8}
\\
&&
\lbrack \mathcal W_{k+1}, {\mathcal {\tilde G}}_{\ell+1}\rbrack+
\lbrack {\mathcal {\tilde G}}_{k+1}, \mathcal W_{\ell+1}\rbrack= 0,
\label{eq:3p9}
\\
&&
\lbrack \mathcal G_{k+1}, \mathcal G_{\ell+1}\rbrack=0,
\qquad 
\lbrack {\mathcal {\tilde G}}_{k+1}, {\mathcal {\tilde G}}_{\ell+1}\rbrack= 0,
\label{eq:3p10}
\\
&&
\lbrack {\mathcal {\tilde G}}_{k+1}, \mathcal G_{\ell+1}\rbrack+
\lbrack \mathcal G_{k+1}, {\mathcal {\tilde G}}_{\ell+1}\rbrack= 0.
\label{eq:3p11}
\end{eqnarray}
In the above equations $\ell \in \mathbb N$
and $\rho = -(q^2-q^{-2})^2$.
 We are using the notation
 $\lbrack X,Y\rbrack=XY-YX$ and $\lbrack X,Y\rbrack_q=
qXY-q^{-1}YX$.
\end{definition}

\noindent 
There is a redundancy among the generators
(\ref{eq:4gens}),
since we could use (\ref{eq:3p1}) to eliminate 
$\lbrace \mathcal G_{k+1}\rbrace_{k \in \mathbb N}$
or 
$\lbrace {\mathcal {\tilde  G}}_{k+1} \rbrace_{k \in \mathbb N}$
in 
(\ref{eq:3p2})--(\ref{eq:3p11}).
These eliminations yield the 
equations in the next lemma.
\begin{lemma}
\label{lem:elimG}
The following equations hold in $\mathcal A_q$. For
$k \in \mathbb N$,
\begin{align}
&\lbrack \lbrack 
\mathcal W_{k+1}, \mathcal W_0\rbrack_q , 
\mathcal W_0 \rbrack 
=  \lbrack
{\mathcal G}_{k+1},
\mathcal W_0 
\rbrack,
\label{eq:elim1}
\\
&\lbrack
\mathcal W_1, \lbrack \mathcal W_1, \mathcal W_{-k}\rbrack_q
\rbrack  =  \lbrack 
\mathcal W_1, {\mathcal G}_{k+1}
\rbrack,
\label{eq:elim2}
\\
&\lbrack \mathcal W_0, \lbrack \mathcal W_0, \mathcal W_{k+1}\rbrack_q \rbrack 
=  \lbrack \mathcal W_0, {\mathcal {\tilde G}}_{k+1} \rbrack,
\label{eq:elim3}
\\
&\lbrack \lbrack 
\mathcal W_{-k}, \mathcal W_1\rbrack_q , 
\mathcal W_1  \rbrack 
=  \lbrack
{\mathcal {\tilde G}}_{k+1},
\mathcal W_1 
\rbrack.
\label{eq:elim4}
\end{align}
For $k,\ell \in \mathbb N$,
\begin{align}
&\lbrack
\lbrack \mathcal W_0, \mathcal W_{k+1} \rbrack,
\lbrack \mathcal W_0, \mathcal W_{\ell+1} \rbrack \rbrack = 0,
\\
&\lbrack
\lbrack \mathcal W_1, \mathcal W_{-k} \rbrack,
\lbrack \mathcal W_1, \mathcal W_{-\ell} \rbrack \rbrack = 0.
\end{align}
\end{lemma}

\noindent We now consider some automorphisms of $\mathcal A_q$.

\begin{lemma} 
\label{lem:PhiAq}
{\rm (See 
\cite[Remarks~1,~2]{basBel}.)}
There exists an automorphism $\Omega$ of
$\mathcal A_q$ that sends
$\mathcal W_{-k} \leftrightarrow \mathcal W_{k+1}$ and
$\mathcal G_{k+1} \leftrightarrow {\mathcal {\tilde G}}_{k+1}$ 
for $k \in \mathbb N$.
Moreover $\Omega $ fixes $\Delta_{k+1}$ for $k \in \mathbb N$.
We have $\Omega^2=1$.
\end{lemma}

\begin{lemma} 
\label{lem:T0Aq}
{\rm (See \cite[Proposition~7.4]{lusztigaut}.)}
\label{conj:W}
There exists an automorphism $T_0$ of 
the algebra $\mathcal A_q$ that acts as follows.
For $k \in \mathbb N$,  $T_0$ sends
\begin{align*}
\mathcal W_{-k} &\mapsto \mathcal W_{-k},
\\
\mathcal W_{k+1} &\mapsto \mathcal W_{k+1} + \frac{q \mathcal W_0^2 \mathcal W_{k+1} -(q+q^{-1})\mathcal W_0 \mathcal W_{k+1}\mathcal W_0+ q^{-1}\mathcal W_{k+1}\mathcal W_0^2}{(q-q^{-1})(q^2-q^{-2})},
\\
\mathcal G_{k+1} &\mapsto \mathcal G_{k+1} + \frac{q \mathcal W_0^2 
\mathcal G_{k+1} -(q+q^{-1})\mathcal W_0 \mathcal G_{k+1}\mathcal W_0+ q^{-1}\mathcal G_{k+1}
\mathcal W_0^2}{(q-q^{-1})(q^2-q^{-2})} = \mathcal{\tilde G}_{k+1},
\\
\mathcal {\tilde G}_{k+1} &\mapsto \mathcal {\tilde G}_{k+1} + \frac{q \mathcal  W_0^2 \mathcal {\tilde G}_{k+1} -
(q+q^{-1})\mathcal W_0\mathcal {\tilde G}_{k+1}\mathcal W_0+ q^{-1}\mathcal {\tilde G}_{k+1}\mathcal W_0^2}{(q-q^{-1})(q^2-q^{-2})},
\\
\Delta_{k+1} &\mapsto \Delta_{k+1}.
\end{align*} 
Moreover $T_0^{-1}$ sends
\begin{align*}
\mathcal W_{-k} &\mapsto \mathcal W_{-k},
\\
\mathcal W_{k+1} &\mapsto \mathcal W_{k+1} + \frac{q^{-1} \mathcal W_0^2\mathcal W_{k+1} -(q+q^{-1})\mathcal W_0 \mathcal W_{k+1}\mathcal W_0+ q \mathcal 
W_{k+1}\mathcal W_0^2}{(q-q^{-1})(q^2-q^{-2})},
\\
\mathcal G_{k+1} &\mapsto \mathcal G_{k+1} + \frac{q^{-1} \mathcal W_0^2
\mathcal G_{k+1} -(q+q^{-1})\mathcal W_0\mathcal G_{k+1}\mathcal W_0+
q\mathcal G_{k+1}\mathcal W_0^2}{(q-q^{-1})(q^2-q^{-2})},
\\
\mathcal {\tilde G}_{k+1} &\mapsto \mathcal{\tilde G}_{k+1} + 
\frac{q^{-1} \mathcal W_0^2 \mathcal {\tilde G}_{k+1} -(q+q^{-1})
\mathcal W_0\mathcal{\tilde G}_{k+1}\mathcal W_0+ q\mathcal {\tilde G}_{k+1}
\mathcal W_0^2}{(q-q^{-1})(q^2-q^{-2})}
= \mathcal G_{k+1},
\\
\Delta_{k+1} &\mapsto \Delta_{k+1}.
\end{align*} 
\end{lemma}

\begin{definition}
\label{def:Phi}
\rm  
Define $T_1=\Omega T_0 \Omega $, where 
$\Omega $ is from
Lemma
\ref{lem:PhiAq}
and
$T_0$ is from Lemma 
\ref{lem:T0Aq}.
 By construction $T_1$ is
an automorphism of the algebra $\mathcal A_q$.
\end{definition}

\begin{lemma} 
\label{conj:Ws}
For $k \in \mathbb N$,  $T_1$ sends
\begin{align*}
\mathcal W_{k+1} &\mapsto \mathcal W_{k+1},
\\
\mathcal W_{-k} &\mapsto \mathcal W_{-k} + \frac{q \mathcal W_1^2 
\mathcal W_{-k} -(q+q^{-1})\mathcal W_1 \mathcal W_{-k}\mathcal W_1+
q^{-1}\mathcal W_{-k}\mathcal W_1^2}{(q-q^{-1})(q^2-q^{-2})},
\\
\mathcal G_{k+1} &\mapsto \mathcal G_{k+1} + \frac{q \mathcal W_1^2 
\mathcal G_{k+1} -(q+q^{-1})\mathcal W_1 \mathcal G_{k+1}\mathcal W_1+ q^{-1}\mathcal G_{k+1}
\mathcal W_1^2}{(q-q^{-1})(q^2-q^{-2})},
\\
\mathcal {\tilde G}_{k+1} &\mapsto \mathcal {\tilde G}_{k+1} +
\frac{q \mathcal  W_1^2 \mathcal {\tilde G}_{k+1} -
(q+q^{-1})\mathcal W_1\mathcal {\tilde G}_{k+1}\mathcal W_1
+ q^{-1}\mathcal {\tilde G}_{k+1}\mathcal W_1^2}{(q-q^{-1})(q^2-q^{-2})}
=\mathcal{G}_{k+1},
\\
\Delta_{k+1} &\mapsto \Delta_{k+1}.
\end{align*} 
Moreover $T_1^{-1}$ sends
\begin{align*}
\mathcal W_{k+1} &\mapsto \mathcal W_{k+1},
\\
\mathcal W_{-k} &\mapsto \mathcal W_{-k} + 
\frac{q^{-1} \mathcal W_1^2\mathcal W_{-k} -
(q+q^{-1})\mathcal W_1 \mathcal W_{-k}\mathcal W_1+ q \mathcal 
W_{-k}\mathcal W_1^2}{(q-q^{-1})(q^2-q^{-2})},
\\
\mathcal G_{k+1} &\mapsto \mathcal G_{k+1} +
\frac{q^{-1} \mathcal W_1^2
\mathcal G_{k+1} -(q+q^{-1})\mathcal W_1\mathcal G_{k+1}\mathcal W_1+
q\mathcal G_{k+1}\mathcal W_1^2}{(q-q^{-1})(q^2-q^{-2})}
= \mathcal {\tilde G}_{k+1},
\\
\mathcal {\tilde G}_{k+1} &\mapsto \mathcal{\tilde G}_{k+1} + 
\frac{q^{-1} \mathcal W_1^2 \mathcal {\tilde G}_{k+1} -(q+q^{-1})
\mathcal W_1\mathcal{\tilde G}_{k+1}\mathcal W_1+
q\mathcal {\tilde G}_{k+1}
\mathcal W_1^2}{(q-q^{-1})(q^2-q^{-2})},
\\
\Delta_{k+1} &\mapsto \Delta_{k+1}.
\end{align*} 
\end{lemma}
\begin{proof}  Use Lemma
\ref{conj:W}
and Definition
\ref{def:Phi}.
\end{proof}

\noindent We have been discussing automorphisms of $\mathcal A_q$.
We now consider antiautomorphisms of $\mathcal A_q$.

\begin{lemma}
\label{lem:S2}
There exists an antiautomorphism $S$ of
$\mathcal A_q$ that sends
\begin{align*}
\mathcal W_{-k} \mapsto \mathcal W_{-k},
\qquad 
\mathcal W_{k+1} \mapsto \mathcal W_{k+1},
\qquad 
\mathcal G_{k+1} \mapsto {\mathcal {\tilde G}}_{k+1},
\qquad 
{\mathcal {\tilde G}}_{k+1} \mapsto \mathcal G_{k+1}
\end{align*}
For $k \in \mathbb N$.
Moreover $S$ fixes $\Delta_{k+1}$ for $k \in \mathbb N$.
We have $S^2=1$.
\end{lemma}
\begin{proof} The antiautomorphism $S$ exists by the form of the defining
relations 
(\ref{eq:3p1})--(\ref{eq:3p11})
for $\mathcal A_q$.
The map $S^2$ is an automorphism of $\mathcal A_q$ that fixes
 $\mathcal W_{-k}$, $\mathcal W_{k+1}$,  $\mathcal G_{k+1}$,
 $\mathcal {\tilde G}_{k+1}$ for $k \in \mathbb N$.
 These elements generate $\mathcal A_q$, so $S^2=1$. For 
$k \in \mathbb N$ the map $S$ fixes $\Delta_{k+1}$
by the form of $\Delta_{k+1}$ given in
\cite[Lemma~2.1]{basBel}.
\end{proof}

\begin{lemma} 
\label{lem:ST2}
For the algebra $\mathcal A_q$,
\begin{align}
S T_0 S = T^{-1}_0, \qquad \qquad
S T_1 S = T^{-1}_1.
\label{eq:gam2}
\end{align}
\end{lemma}
\begin{proof} Similar to the proof of
Lemma \ref{lem:alphabeta1}.
\end{proof}

\noindent We now obtain our second main result.
Recall the free product 
$\mathbb Z_2 \star \mathbb Z_2 \star \mathbb Z_2$
from above Lemma
\ref{lem:isoiso}.
For $k \in \mathbb N$ define
\begin{align}
\mathcal W'_{-k} = \mathcal W_{k+1}, \qquad \qquad
\mathcal W''_{-k} = -\frac{{\mathcal{ \tilde G}}_{k+1}}{ q^2-q^{-2}}.
\label{eq:prime}
\end{align}
\noindent Note by 
(\ref{eq:3p1}),
(\ref{eq:3p7}),
(\ref{eq:3p9})
that
\begin{align}
\lbrack \mathcal W_{-k}, \mathcal W'_0\rbrack = 
\lbrack \mathcal W_{0}, \mathcal W'_{-k}\rbrack,
\qquad 
\lbrack \mathcal W'_{-k}, \mathcal W''_0\rbrack = 
\lbrack \mathcal W'_{0}, \mathcal W''_{-k}\rbrack,
\qquad 
\lbrack \mathcal W''_{-k}, \mathcal W_0\rbrack = 
\lbrack \mathcal  W''_{0}, \mathcal W_{-k}\rbrack.
\label{eq:helpful}
\end{align}

\begin{theorem} 
\label{thm:secondres}
The free product $\mathbb Z_2 \star \mathbb Z_2 \star \mathbb Z_2$
acts on the algebra 
$\mathcal A_q$ as a group of (auto/antiauto)-morphisms
in the following way.
\begin{enumerate}
\item[\rm (i)] The generator $a$ acts as an antiautomorphism 
 that sends
\begin{align*}
\mathcal W_{-k} &\mapsto \mathcal W_{-k} + \frac{\lbrack \mathcal W'_{-k}, \mathcal W''_0\rbrack}{q-q^{-1}},
\\
\mathcal W'_{-k} &\mapsto \mathcal W'_{-k},
\\
\mathcal W''_{-k} &\mapsto \mathcal W''_{-k},
\\
\Delta_{k+1} &\mapsto \Delta_{k+1}.
\end{align*}
\item[\rm (ii)] The generator $b$ acts as an antiautomorphism 
that sends
\begin{align*}
\mathcal W'_{-k} &\mapsto \mathcal W'_{-k} + \frac{\lbrack \mathcal W''_{-k},
\mathcal W_0\rbrack}{q-q^{-1}},
\\
\mathcal W''_{-k} &\mapsto \mathcal W''_{-k},
\\
\mathcal W_{-k} &\mapsto \mathcal W_{-k},
\\
\Delta_{k+1} &\mapsto \Delta_{k+1}.
\end{align*}
\item[\rm (iii)] The generator $c$ acts as an antiautomorphism that 
sends
\begin{align*}
\mathcal W''_{-k} &\mapsto \mathcal W''_{-k} + \frac{\lbrack \mathcal W_{-k},
\mathcal W'_0\rbrack}{q-q^{-1}},
\\
\mathcal W_{-k} &\mapsto \mathcal W_{-k},
\\
\mathcal W'_{-k} &\mapsto \mathcal W'_{-k},
\\
\Delta_{k+1} &\mapsto \Delta_{k+1}.
\end{align*}
\item[\rm (iv)] On $\mathcal A_q$,
\begin{align}
&a= ST_1=T^{-1}_1 S, \qquad \qquad b= T_0 S=ST^{-1}_0, \qquad \qquad c=S,
\label{eq:Part1}
\\
&T_0 = bc, \;\quad \qquad  
T^{-1}_0 = cb, \;\quad \qquad 
T_1 = ca, \;\quad \qquad 
T^{-1}_1 = ac.
\label{eq:Part2}
\end{align}
%\item[\rm (iv)]
%The above
%$\mathbb Z_2 \star \mathbb Z_2 \star \mathbb Z_2$ action is
%faithful.
\end{enumerate}
\end{theorem}
\begin{proof} 
For notational convenience we identify the group
$\mathbb Z_2 \star \mathbb Z_2 \star \mathbb Z_2$ with the
group defined in Lemma
\ref{lem:isoiso}, via the isomorphism in
Lemma
\ref{lem:isoiso}.
Comparing 
(\ref{eq:rels}) 
with the relations in
Lemmas
\ref{lem:S2},
\ref{lem:ST2}, we obtain a group 
homomorphism 
$\mathbb Z_2 \star \mathbb Z_2 \star \mathbb Z_2 \to {\rm AAut}(\mathcal A_q)$
that sends 
$s\mapsto S$ and
$t^{\pm 1}_0 \mapsto T^{\pm 1}_0 $ and $t^{\pm 1}_1 \mapsto T^{\pm 1}_1$.
This group homomorphism gives
an action of 
$\mathbb Z_2 \star \mathbb Z_2 \star \mathbb Z_2$
on 
the algebra $\mathcal A_q$ as a group of
(auto/antiauto)-morphisms such that
$s$,
$t^{\pm 1}_0$,
$t^{\pm 1}_1$
act as $S$, 
$T^{\pm 1}_0$,
$T^{\pm 1}_1$,
respectively. Using the identifications
(\ref{eq:send1}),
(\ref{eq:send2})
we find that this action satisfies condition (iv) in
the theorem statement.
By
(\ref{eq:Part1}) each of $a$, $b$, $c$ acts on $\mathcal A_q$
as an antiautomorphism. For these elements the action on
$\mathcal W_{-k}$,
$\mathcal W'_{-k}$,
$\mathcal W''_{-k}$,
$\Delta_{k+1}$
is routinely obtained using
Lemmas
\ref{lem:T0Aq},
\ref{conj:Ws},
\ref{lem:S2}
along with Lemma
\ref{lem:elimG} and
(\ref{eq:helpful}).
%Lemmas
%\ref{conj:W},
%\ref{conj:Ws}.
\end{proof}

\section{Suggestions for future research}

\noindent 
In this section we give some conjectures and problems concerning
$\mathcal O_q$ and $\mathcal A_q$.
\medskip

\noindent 
Earlier in this paper we gave a 
$\mathbb Z_2 \star \mathbb Z_2 \star \mathbb Z_2$ action on
$\mathcal O_q$ and $\mathcal A_q$. It is natural to
ask whether these algebras are characterized by this sort of
$\mathbb Z_2 \star \mathbb Z_2 \star \mathbb Z_2$ action.
As we pursue this question, let us begin with the simpler case of
$\mathcal O_q$. The following concept is motivated by Theorem
\ref{thm:z3}.

\begin{definition}\rm
Let $\mathcal A$ denote an algebra.
A sequence $A,B,C$ of elements in $\mathcal A$
is called a
{\it flipping triple}
whenever:
\begin{enumerate}
\item[\rm (i)] there exists an antiautomorphism of
$\mathcal A$ that sends
\begin{eqnarray*}
A \mapsto A + BC-CB,\qquad \qquad 
B \mapsto B, \qquad \qquad C \mapsto C;
\end{eqnarray*}
\item[\rm (ii)] there exists an antiautomorphism of
$\mathcal A$ that sends
\begin{eqnarray*}
B \mapsto B + CA-AC, \qquad \qquad 
C \mapsto C, \qquad \qquad A \mapsto A;
\end{eqnarray*}
\item[\rm (iii)] there exists an antiautomorphism of
$\mathcal A$ that sends
\begin{eqnarray*}
C \mapsto C + A B-BA,\qquad \qquad 
A \mapsto A, \qquad \qquad B \mapsto B;
\end{eqnarray*}
\item[\rm (iv)] the algebra $\mathcal A$ is generated by
$A,B,C$.
\end{enumerate}
\end{definition}

\begin{example}
\rm
Recall from Definition
\ref{def:qons}
the generators $A, B$ for the $q$-Onsager algebra $\mathcal O_q$.
Recall the element $C$ from
(\ref{eq:C}).
 By Theorem
\ref{thm:z3}
the sequence $A/(q-q^{-1})$, $B/(q-q^{-1})$, $C/(q-q^{-1})$
is a flipping triple for
 $\mathcal O_q$.
\end{example}

\begin{example}
\rm
Assume that $A,B,C$ freely generate $\mathcal A$.
One routinely checks that
$A,B,C$ is a flipping triple for $\mathcal A$.
\end{example}

\begin{problem}\rm
Find all the sequences 
$A$, $B$, $C$, $\mathcal A$ such that
$A$, $B$, $C$ is a flipping triple in the algebra $\mathcal A$.
\end{problem}

\noindent We define some notation.
Let $\lambda_1,\lambda_2,\ldots$ denote mutually
commuting indeterminates. Let $\mathbb F\lbrack \lambda_1, \lambda_2,\ldots
\rbrack$ denote the algebra  of polynomials in
 $\lambda_1,\lambda_2,\ldots$ 
that have all coefficients
in $\mathbb F$.
For a subset $Y \subseteq \mathcal A_q$ let $\langle Y\rangle$
denote the subalgebra of $\mathcal A_q$ generated by $Y$.
Shortly we will encounter some tensor products. All tensor products
in this paper are understood to be over $\mathbb F$.
\medskip

\noindent The following conjecture about $\mathcal A_q$
is a variation on
\cite[Conjecture~1]{basBel}.
\begin{conjecture}
\label{conj:4}
\rm The following {\rm (i)--(v)} hold:
\begin{enumerate}
\item[\rm (i)]
there exists an algebra isomorphism
 $\mathbb F\lbrack \lambda_1, \lambda_2,\ldots \rbrack \to 
\langle \mathcal W_0, \mathcal W_{-1}, \ldots \rangle$ 
 that sends $\lambda_{k+1} \mapsto \mathcal W_{-k}$ for $k \in \mathbb N$;
\item[\rm (ii)]
there exists an algebra isomorphism
 $\mathbb F\lbrack \lambda_1, \lambda_2,\ldots \rbrack \to 
\langle \mathcal W_1,\mathcal W_2,\ldots  \rangle$ 
 that sends $\lambda_{k+1} \mapsto \mathcal W_{k+1}$ for $k \in \mathbb N$;
\item[\rm (iii)]
there exists an algebra isomorphism
 $\mathbb F\lbrack \lambda_1, \lambda_2,\ldots \rbrack \to 
\langle \mathcal G_1,\mathcal G_2,\ldots \rangle$ 
 that sends $\lambda_{k+1} \mapsto \mathcal G_{k+1}$ for $k \in \mathbb N$;
\item[\rm (iv)]
there exists an algebra isomorphism
 $\mathbb F\lbrack \lambda_1, \lambda_2,\ldots \rbrack \to 
\langle \mathcal {\tilde G}_1,\mathcal {\tilde G}_2,\ldots \rangle$ 
 that sends $\lambda_{k+1} \mapsto \mathcal {\tilde G}_{k+1}$ for $k \in \mathbb N$;
\item[\rm (v)] the multiplication map
\begin{align*}
\langle  \mathcal W_0, \mathcal W{-1},\ldots \rangle
\otimes 
\langle  \mathcal G_1,\mathcal G_2,\ldots \rangle
\otimes 
\langle \mathcal {\tilde G}_1,\mathcal {\tilde G}_2,\ldots  \rangle
\otimes
\langle  \mathcal W_1,\mathcal W_2,\ldots \rangle
 & \to   \mathcal A_q
\\
 u \otimes v \otimes w \otimes x  &\mapsto  uvwx
 \end{align*}
is an isomorphism of vector spaces.
\end{enumerate}
\end{conjecture}
\noindent A proof of Conjecture
\ref{conj:4} would yield a PBW basis for $\mathcal A_q$.
\medskip

\noindent The next conjecture concerns the center 
 $Z(\mathcal A_q)$.

%%%%%%%%%%%%%%%%%%%%
%\begin{conjecture}\rm The following {\rm (i)--(iii)} are the same:
%\begin{enumerate}
%\item[\rm (i)]
%the center $Z(\mathcal A_q)$;
%\item[\rm (ii)] 
%the subalgebra of $\mathcal A_q$
% generated by
%$\lbrace \Delta_{k+1} \rbrace_{k \in \mathbb N}$;
%\item[\rm (iii)] 
%the set of elements in $\mathcal A_q$ that are fixed by $T_0$ and $T_1$.
%\end{enumerate}
%\end{conjecture}
%%%%%%%%%%%%%%%%%%%

%%% move
%\noindent Let $J$ denote the 2-sided ideal of $\mathcal A_q$
%generated by $\mathcal W_0, \mathcal W_1$. 
%\begin{conjecture}\rm The sum
%$\mathcal A_q = J + Z(\mathcal A_q)$
%is direct.
%\end{conjecture}
%%%%%%%%%%%%%%%

\begin{conjecture}
\label{conj:3}
\rm The following {\rm (i)--(iii)} hold:
\begin{enumerate}
\item[\rm (i)]
there exists an algebra isomorphism
 $\mathbb F\lbrack \lambda_1, \lambda_2,\ldots \rbrack \to Z(\mathcal A_q)$
 that sends $\lambda_{k+1} \mapsto \Delta_{k+1}$ for $k \in \mathbb N$;
\item[\rm (ii)]
there exists an algebra isomorphism
$\mathcal O_q \to \langle \mathcal W_0, \mathcal W_1\rangle$ 
that sends $A\mapsto \mathcal W_0$ and
$B\mapsto \mathcal W_1$;
\item[\rm (iii)] the multiplication map 
\begin{align*}
\langle \mathcal W_0, \mathcal W_1\rangle 
\otimes
Z(\mathcal A_q) 
 & \to   \mathcal A_q
\\
 u \otimes v  &\mapsto  uv
 \end{align*}
 is an isomorphsim of algebras.
\end{enumerate}
 \end{conjecture}

\noindent A proof
of Conjecture \ref{conj:3} would yield an algebra
isomorphism
$
\mathcal O_q \otimes \mathbb F \lbrack \lambda_1,
\lambda_2,\ldots \rbrack
\to
\mathcal A_q 
$.
\medskip

\noindent Above Lemma
\ref{lem:elimG}
we mentioned a redundancy among the generators
(\ref{eq:4gens}) of $\mathcal A_q$.
We now pursue this theme more deeply.
Using
(\ref{eq:3p1}) we 
eliminate the generators
$\lbrace \mathcal G_{k+1}\rbrace_{k \in \mathbb N}$:
\begin{align*}
\mathcal G_{k+1} = \mathcal {\tilde G}_{k+1} + (q+q^{-1}) 
\lbrack \mathcal W_1, \mathcal W_{-k} \rbrack
\qquad \qquad (k \in \mathbb N).
%%\label{eq:getrid}
\end{align*}
Next we use
(\ref{eq:3p2}), (\ref{eq:3p3}) to recursively eliminate 
$\mathcal W_{-k}$, $\mathcal W_{k+1}$ for $k\geq 1$:
\begin{align*}
\mathcal W_{-1} &= \mathcal W_1 -\frac{\lbrack \mathcal {\tilde G}_{1}, 
\mathcal W_0\rbrack_q}{(q^2-q^{-2})^2},
\\
\mathcal W_{3} &= \mathcal W_1 
-
\frac{\lbrack \mathcal {\tilde G}_{1}, 
\mathcal W_0\rbrack_q}{(q^2-q^{-2})^2}
-
\frac{\lbrack
\mathcal W_1,
\mathcal {\tilde G}_{2} 
\rbrack_q}{(q^2-q^{-2})^2},
\\
\mathcal W_{-3} &= \mathcal W_1 
-
\frac{\lbrack \mathcal {\tilde G}_{1}, 
\mathcal W_0\rbrack_q}{(q^2-q^{-2})^2}
-
\frac{\lbrack
\mathcal W_1,
\mathcal {\tilde G}_{2} 
\rbrack_q}{(q^2-q^{-2})^2}
-
\frac{\lbrack \mathcal {\tilde G}_{3}, 
\mathcal W_0\rbrack_q}{(q^2-q^{-2})^2},
\\
\mathcal W_{5} &= \mathcal W_1 
-
\frac{\lbrack \mathcal {\tilde G}_{1}, 
\mathcal W_0\rbrack_q}{(q^2-q^{-2})^2}
-
\frac{\lbrack
\mathcal W_1,
\mathcal {\tilde G}_{2} 
\rbrack_q}{(q^2-q^{-2})^2}
-
\frac{\lbrack \mathcal {\tilde G}_{3}, 
\mathcal W_0\rbrack_q}{(q^2-q^{-2})^2}
-
\frac{\lbrack
\mathcal W_1,
\mathcal {\tilde G}_{4} 
\rbrack_q}{(q^2-q^{-2})^2},
\\
\mathcal W_{-5} &= \mathcal W_1 
-
\frac{\lbrack \mathcal {\tilde G}_{1}, 
\mathcal W_0\rbrack_q}{(q^2-q^{-2})^2}
-
\frac{\lbrack
\mathcal W_1,
\mathcal {\tilde G}_{2} 
\rbrack_q}{(q^2-q^{-2})^2}
-
\frac{\lbrack \mathcal {\tilde G}_{3}, 
\mathcal W_0\rbrack_q}{(q^2-q^{-2})^2}
-
\frac{\lbrack
\mathcal W_1,
\mathcal {\tilde G}_{4} 
\rbrack_q}{(q^2-q^{-2})^2}
-
\frac{\lbrack \mathcal {\tilde G}_{5}, 
\mathcal W_0\rbrack_q}{(q^2-q^{-2})^2},
\\
& \cdots
\end{align*}
\begin{align*}
\mathcal W_2 &= \mathcal W_0 -
\frac{
\lbrack \mathcal W_1, \mathcal {\tilde G}_1\rbrack_q}{(q^2-q^{-2})^2
},
\\
\mathcal W_{-2} &= \mathcal W_0 -
\frac{
\lbrack \mathcal W_1, \mathcal {\tilde G}_1\rbrack_q}{(q^2-q^{-2})^2}
-
\frac{
\lbrack \mathcal {\tilde G}_2, \mathcal W_0\rbrack_q}{(q^2-q^{-2})^2},
\\
\mathcal W_{4} &= \mathcal W_0 -
\frac{
\lbrack \mathcal W_1, \mathcal {\tilde G}_1\rbrack_q}{(q^2-q^{-2})^2}
-
\frac{
\lbrack \mathcal {\tilde G}_2, \mathcal W_0\rbrack_q}{(q^2-q^{-2})^2}
-
\frac{
\lbrack \mathcal W_1, \mathcal {\tilde G}_3\rbrack_q}{(q^2-q^{-2})^2},
\\
\mathcal W_{-4} &= \mathcal W_0 -
\frac{
\lbrack \mathcal W_1, \mathcal {\tilde G}_1\rbrack_q}{(q^2-q^{-2})^2}
-
\frac{
\lbrack \mathcal {\tilde G}_2, \mathcal W_0\rbrack_q}{(q^2-q^{-2})^2}
-
\frac{
\lbrack \mathcal W_1, \mathcal {\tilde G}_3\rbrack_q}{(q^2-q^{-2})^2}
-
\frac{
\lbrack \mathcal {\tilde G}_4, \mathcal W_0\rbrack_q}{(q^2-q^{-2})^2},
\\
\mathcal W_{6} &= \mathcal W_0 -
\frac{
\lbrack \mathcal W_1, \mathcal {\tilde G}_1\rbrack_q}{(q^2-q^{-2})^2}
-
\frac{
\lbrack \mathcal {\tilde G}_2, \mathcal W_0\rbrack_q}{(q^2-q^{-2})^2}
-
\frac{
\lbrack \mathcal W_1, \mathcal {\tilde G}_3\rbrack_q}{(q^2-q^{-2})^2}
-
\frac{
\lbrack \mathcal {\tilde G}_4, \mathcal W_0\rbrack_q}{(q^2-q^{-2})^2}
-
\frac{
\lbrack \mathcal W_1, \mathcal {\tilde G}_5\rbrack_q}{(q^2-q^{-2})^2},
\\
&\cdots
\end{align*}

\noindent 
For any integer $k \geq 1$, 
the generators $\mathcal W_{-k}$, $\mathcal W_{k+1}$
are given as follows. 
\\
\noindent For odd $k=2r+1$,
\begin{align*}
\mathcal W_{-k} = \mathcal W_1
- 
\sum_{\ell=0}^r 
\frac{
\lbrack 
 \mathcal {\tilde G}_{2\ell+1}, \mathcal W_0\rbrack_q
}{(q^2-q^{-2})^2} 
-
\sum_{\ell=1}^r 
\frac{
\lbrack
 \mathcal W_1,
 \mathcal {\tilde G}_{2\ell} 
 \rbrack_q
}{(q^2-q^{-2})^2},
\\
\mathcal W_{k+1} = \mathcal W_0
- 
\sum_{\ell=0}^r 
\frac{
\lbrack 
\mathcal W_1,
 \mathcal {\tilde G}_{2\ell+1}\rbrack_q
}{(q^2-q^{-2})^2} 
-
\sum_{\ell=1}^r 
\frac{
\lbrack
 \mathcal {\tilde G}_{2\ell},
 \mathcal W_0
 \rbrack_q
}{(q^2-q^{-2})^2}.
\end{align*}
For even $k=2r$,
\begin{align*}
\mathcal W_{-k} = \mathcal W_0
- 
\sum_{\ell=0}^{r-1} 
\frac{
\lbrack 
\mathcal W_1,
 \mathcal {\tilde G}_{2\ell+1}\rbrack_q
}{(q^2-q^{-2})^2} 
-
\sum_{\ell=1}^r 
\frac{
\lbrack
 \mathcal {\tilde G}_{2\ell},
 \mathcal W_0
 \rbrack_q
}{(q^2-q^{-2})^2},
\\
\mathcal W_{k+1} = \mathcal W_1
- 
\sum_{\ell=0}^{r-1} 
\frac{
\lbrack 
 \mathcal {\tilde G}_{2\ell+1}, \mathcal W_0\rbrack_q
}{(q^2-q^{-2})^2} 
-
\sum_{\ell=1}^r 
\frac{
\lbrack
 \mathcal W_1,
 \mathcal {\tilde G}_{2\ell} 
 \rbrack_q
}{(q^2-q^{-2})^2}. 
\end{align*}

\noindent So far, we have expressed the generators
(\ref{eq:4gens}) 
in terms of 
$\mathcal W_0, \mathcal W_1$,$\lbrace \mathcal {\tilde G}_{k+1}\rbrace_{k\in \mathbb N}$.
We now consider how these remaining generators
are related to each other.

\begin{lemma}
\label{lem:newrels}
The following relations hold in the algebra $\mathcal A_q$:
\begin{align*}
&\lbrack \mathcal W_0, \mathcal {\tilde G}_1 \rbrack = 
\lbrack \mathcal W_0, \lbrack \mathcal W_0, \mathcal W_1 \rbrack_q \rbrack,
\\
&\lbrack \mathcal {\tilde G}_1, \mathcal W_1 \rbrack = 
\lbrack \lbrack \mathcal W_0, \mathcal W_1 \rbrack_q, \mathcal W_1 \rbrack
\end{align*}
and for $k\geq 1$,
\begin{align*}
&\lbrack \mathcal {\tilde G}_{k+1}, \mathcal W_0 \rbrack = 
\frac{
\lbrack \mathcal W_0, \lbrack \mathcal W_0, \lbrack \mathcal W_1,
\mathcal {\tilde G}_k
\rbrack_q
\rbrack_q
\rbrack}{(q^2-q^{-2})^2},
\\
&\lbrack \mathcal W_1, \mathcal {\tilde G}_{k+1}\rbrack = 
\frac{
\lbrack\lbrack \lbrack \mathcal {\tilde G}_k, \mathcal W_0 \rbrack_q, 
\mathcal W_1 \rbrack_q,
\mathcal W_1
\rbrack}{(q^2-q^{-2})^2}.
\end{align*}
\end{lemma}
\begin{proof} The first two relations are 
(\ref{eq:elim3}),
(\ref{eq:elim4}) with $k=0$.
To obtain the third relation, use
(\ref{eq:elim3}) and 
(\ref{eq:3p3}),  
(\ref{eq:3p4})  
to obtain
\begin{align*}
\lbrack \mathcal {\tilde G}_{k+1}, \mathcal W_0 \rbrack 
&= 
- \lbrack \mathcal W_0, \lbrack \mathcal W_0,  \mathcal W_{k+1}\rbrack_q \rbrack
\\
&= 
-\lbrack \mathcal W_0, \lbrack \mathcal W_0,  \mathcal W_{k+1}\rbrack \rbrack_q
\\
&=
-\lbrack \mathcal W_0, \lbrack \mathcal W_0,  
\mathcal W_{k+1}-\mathcal W_{1-k} \rbrack \rbrack_q
\\
&= \frac{
\lbrack \mathcal W_0, \lbrack \mathcal W_0, \lbrack \mathcal W_1,
\mathcal {\tilde G}_k
\rbrack_q
\rbrack
\rbrack_q}{(q^2-q^{-2})^2}
\\
&= \frac{
\lbrack \mathcal W_0, \lbrack \mathcal W_0, \lbrack \mathcal W_1,
\mathcal {\tilde G}_k
\rbrack_q
\rbrack_q
\rbrack}{(q^2-q^{-2})^2}.
\end{align*}
The last relation is similarly obtained.
\end{proof}

\begin{conjecture}
\label{conj:pres}
\rm The algebra $\mathcal A_q$ has
a presentation by generators 
$\mathcal W_0$, 
$\mathcal W_1$, 
$\lbrace \mathcal {\tilde G}_{k+1}\rbrace_{k \in \mathbb N}$ and
relations
\begin{align*}
&\lbrack \mathcal W_0, \lbrack \mathcal W_0, \lbrack \mathcal W_0,
\mathcal W_1 \rbrack_q \rbrack_{q^{-1}} \rbrack = (q^2-q^{-2})^2
\lbrack \mathcal W_1,\mathcal W_0\rbrack,
\\
&\lbrack \mathcal W_1, \lbrack \mathcal W_1, \lbrack \mathcal W_1,
\mathcal W_0 \rbrack_q \rbrack_{q^{-1}} \rbrack = (q^2-q^{-2})^2
\lbrack \mathcal W_0,\mathcal W_1\rbrack,
\\
&\lbrack \mathcal W_0, \mathcal {\tilde G}_1 \rbrack = 
\lbrack \mathcal W_0, \lbrack \mathcal W_0, \mathcal W_1 \rbrack_q \rbrack,
\\
&\lbrack \mathcal {\tilde G}_1, \mathcal W_1 \rbrack = 
\lbrack \lbrack \mathcal W_0, \mathcal W_1 \rbrack_q, \mathcal W_1 \rbrack,
\\
&\lbrack \mathcal {\tilde G}_{k+1}, \mathcal W_0 \rbrack = 
\frac{
\lbrack \mathcal W_0, \lbrack \mathcal W_0, \lbrack \mathcal W_1,
\mathcal {\tilde G}_k
\rbrack_q
\rbrack_q
\rbrack}{(q^2-q^{-2})^2} \qquad \qquad (k\geq 1),
\\
&\lbrack \mathcal W_1, \mathcal {\tilde G}_{k+1}\rbrack = 
\frac{
\lbrack\lbrack \lbrack \mathcal {\tilde G}_k, \mathcal W_0 \rbrack_q, 
\mathcal W_1 \rbrack_q,
\mathcal W_1
\rbrack}{(q^2-q^{-2})^2} \qquad \qquad (k\geq 1),
\\
&\lbrack \mathcal {\tilde G}_{k+1}, \mathcal {\tilde G}_{\ell+1} \rbrack = 0
\qquad \qquad 
(k, \ell \in \mathbb N).
\end{align*}
%In the above relations $k, \ell \in \mathbb N$.
\end{conjecture}

\section{Acknowledgment} 
The author thanks Pascal Baseilhac and
Samuel Belliard for many discussions about the $q$-Onsager algebra
and its current algebra.

\bigskip

\noindent Paul Terwilliger \hfil\break
\noindent Department of Mathematics \hfil\break
\noindent University of Wisconsin \hfil\break
\noindent 480 Lincoln Drive \hfil\break
\noindent Madison, WI 53706-1388 USA \hfil\break
\noindent email: {\tt terwilli@math.wisc.edu }\hfil\break

\end{document}